\newif\ifconference
\newif\ifanonymous
\conferencetrue
\anonymousfalse

\ifconference
\documentclass[10pt,a4paper,final]{article}
\else
\documentclass[11pt,a4paper]{article}
\fi
\title{Near-Universally-Optimal Differentially Private \\Minimum Spanning Trees}

\author{
	\ifanonymous
	Anonymous Authors
	\else
	Richard Hladík\thanks{Supported by the VILLUM Foundation grant 54451. The work was done while this author was visiting BARC at the University of Copenhagen.}\\
    \texttt{rihl@uralyx.cz}\\
	ETH Zurich
    \and
	Jakub Tětek\thanks{Supported by the VILLUM Foundation grant 54451.}\\
	\texttt{j.tetek@gmail.com}\\
	BARC, Univ. of Copenhagen
	\fi
}

\date{}

\ifconference
\usepackage[margin=2cm,lmargin=2cm]{geometry}
\else
\usepackage[margin=1in]{geometry}
\fi
\usepackage[dvipsnames,svgnames]{xcolor}
\usepackage{algpseudocodex}
\usepackage{xspace}
\usepackage[textwidth=20mm,obeyFinal,colorinlistoftodos]{todonotes}
\usepackage{amsmath}        
\usepackage{amsfonts}       
\usepackage{amsthm}         
\usepackage{thmtools}
\usepackage{thm-restate}
\usepackage{amssymb}
\usepackage{bbding}         
\usepackage{bm}             
\usepackage{dsfont}
\usepackage{textcomp}
\usepackage{graphicx}       
\usepackage{fancyvrb}       
\usepackage{paralist}
\usepackage{mathtools}
\usepackage{multirow}
\usepackage{booktabs}
\usepackage{adjustbox}
\newcommand{\mytodo}[2]{\todo[size=\tiny, color=#1!50!white]{#2}\xspace}
\newcommand{\myinlinetodo}[2]{\todo[size=\small, color=#1!50!white, inline]{#2}\xspace}
\newcommand{\rh}[1]{\mytodo{green}{RH: #1}}
\newcommand{\rhinline}[1]{\myinlinetodo{green}{RH: #1}}
\setuptodonotes{size=\small}
\usepackage[backend=biber,style=alphabetic,sorting=nty]{biblatex}         
\usepackage{hyperref}
\hypersetup{
    unicode=true,          
    colorlinks=true,        
    linkcolor=red,          
    citecolor=DarkGreen,        
    filecolor=magenta,      
    urlcolor=cyan           
}
\usepackage{enumitem}
\usepackage{graphicx} 
\usepackage{appendix}
\usepackage[capitalize, nameinlink,noabbrev]{cleveref}
\usepackage[left,mathlines]{lineno}

\makeatletter
\AtBeginDocument{
  \hypersetup{
    pdftitle = {\@title},
    pdfauthor = {\ourauthors},
	pdfsubject = {\ourabstract},
	pdfkeywords = {\ourkeywords},
  }
}
\makeatother

\ifanonymous
\linenumbers
\fi

\usepackage{setspace}

\renewcommand\O{\mathcal{O}}

\theoremstyle{plain}
\newtheorem{theorem}{Theorem}[section]
\newtheorem{lemma}[theorem]{Lemma}
\newtheorem{meta-theorem}[theorem]{Meta-Theorem}
\newtheorem{claim}[theorem]{Claim}
\newtheorem{corollary}[theorem]{Corollary}

\newtheorem{fact}[theorem]{Fact}

\theoremstyle{definition}
\newtheorem{definition}[theorem]{Definition}

\newtheorem{algorithm}[theorem]{Algorithm}

\theoremstyle{remark}

\crefname{theorem}{Theorem}{Theorems}
\crefname{proposition}{Proposition}{Propositions}
\crefname{observation}{Observation}{Observations}
\crefname{lemma}{Lemma}{Lemmas}
\crefname{claim}{Claim}{Claims}
\crefname{problem}{Problem}{Problems}
\crefname{conjecture}{Conjecture}{Conjectures}
\crefname{question}{Question}{Questions}
\crefname{example}{Example}{Examples}
\crefname{fact}{Fact}{Facts}    
\crefname{invariant}{Invariant}{Invariants}
\crefname{rule}{Rule}{Rules}

\let\cref=\Cref
\let\citet\textcite

\newcommand\N{\mathbb{N}}
\newcommand\R{\mathbb{R}}

\newcommand\vew{\mathbf{w}}
\newcommand\vex{\mathbf{x}}
\newcommand\vey{\mathbf{y}}

\newcommand\simone{\sim_1}
\newcommand\siminf{\sim_\infty}
\newcommand\eps{\varepsilon}
\newcommand\trees{\mathcal T}

\DeclareMathOperator*\E{\mathbb{E}}
\newcommand\ham{d_H}

\newcommand\algo{\mathcal A}
\newcommand\oneweights{\mathds{1}}
\DeclareMathOperator*{\Lap}{Lap}
\DeclareMathOperator*{\diam}{diam_\mathcal T}
\newcommand\W{\mathcal W}
\newcommand\light{\mathcal L}
\newcommand\X{\mathcal X}
\newcommand\Y{\mathcal Y}
\newcommand\C{\mathcal C}

\graphicspath{{fig/}}

\ifanonymous
\def\ourauthors{Anonymous authors}
\else
\def\ourauthors{Richard Hladík, Jakub Tětek}
\fi

\def\ourabstract{%
Devising mechanisms with good beyond-worst-case input-dependent performance has been an important focus of differential privacy, with techniques such as smooth sensitivity, propose-test-release, or inverse sensitivity mechanism being developed to achieve this goal. This makes it very natural to use the notion of universal optimality in differential privacy. Universal optimality is a strong instance-specific optimality guarantee for problems on weighted graphs, which roughly states that for any fixed underlying (unweighted) graph, the algorithm is optimal in the worst-case sense, with respect to the possible setting of the edge weights. 
\texorpdfstring{\par}{}%
	In this paper, we give the first such result in differential privacy. Namely, we prove that a simple differentially private mechanism for approximately releasing the minimum spanning tree is near-optimal in the sense of universal optimality for the \texorpdfstring{$\ell_1$}{ℓ₁} neighbor relation. 
 Previously, it was only known that this mechanism is nearly optimal in the worst case. We then focus on the \texorpdfstring{$\ell_\infty$}{ℓ∞} neighbor relation, for which the described mechanism is not optimal. We show that one may implement the exponential mechanism for MST in polynomial time, and that this results in universal near-optimality for both the \texorpdfstring{$\ell_1$}{ℓ₁} and the \texorpdfstring{$\ell_\infty$}{ℓ∞} neighbor relations.
}

\def\ourkeywords{differential privacy, universal optimality, minimum spanning trees}

\begin{document}

\maketitle

\begin{abstract}
\ourabstract
\end{abstract}

 \thispagestyle{empty}
\newpage
\clearpage
\setcounter{page}{1}

\section{Introduction}

The minimum spanning tree (MST) problem is one of the classic combinatorial problems, making the release of an MST a fundamental question in differential privacy.
Unfortunately, if the edge set is private, this problem cannot be solved, as it would require us to release a subset of the edges (note that we want to release the edges of the MST and not just their total weight).
It is therefore natural to consider this problem in the (standard) setting where the underlying unweighted graph is public, but the weights are private. 

The following simple near-linear-time mechanism has been proposed for this problem \cite{mst-laplace}: add Laplacian noise to the edge weights, making them private, then find the MST with the noisy weights. 
At the same time, the author shows that this mechanism is near-optimal in the worst case.

In this paper, we prove a much stronger optimality claim for this algorithm, namely that it is universally optimal up to a logarithmic factor for the $\ell_1$ neighbor relation (i.e., two graphs on the same edge set are neighboring if the edge weights differ by $\leq 1$ in the $\ell_1$ norm).
%
An algorithm (or differentially private mechanism) that works on weighted graphs is said to be universally optimal, if for any fixed underlying unweighted graph $G$, it is worst-case optimal with respect to the possible settings of edge weights $\vew$. That is, if we write a weighted graph $\Tilde{G} = (G, \vew)$, universal optimality states that for any fixed $G$, we are worst-case optimal w.r.t.~$\vew$. Standard worst-case optimality, on the other hand, is worst-case w.r.t.~both $\vew$ and $G$.

We then focus on differential privacy with the $\ell_\infty$ neighbor relation (defined analogously to the $\ell_1$ neighbor relation), for which the above algorithm is \emph{not} universally optimal. We instead prove that one may implement the exponential mechanism for MST in polynomial time by relying on a known sampling result \cite{mst-sampling-in-matrix-multiplication}. We prove that this more complicated and somewhat slower $\O(n^\omega)$-time algorithm does achieve universal optimality up to a logarithmic factor for both the $\ell_1$ and the $\ell_\infty$ neighbor relations. For the $\ell_\infty$ neighbor relation, this improves upon the PAMST algorithm of \citet{mst-pamst} which is only known to be near-optimal in the worst-case sense.

Our results are the first to prove the universal optimality of a differentially private mechanism.\footnote{The name universal optimality unfortunately has different meaning in different contexts. There are several results that are universally optimal with one of these different meanings, as we discuss in \Cref{sec:related_work}. We stick with the meaning commonly used in distributed algorithms.} 
This is despite the fact that previous work in differential privacy has put a lot of emphasis on instance-specific performance guarantees: smooth sensitivity, propose-test-release, privately bounding local sensitivity, and inverse sensitivity mechanism are all examples of this trend. This makes it very natural to focus on universal optimality, perhaps making it somewhat surprising that universal optimality is not already being commonly used in differential privacy.

As a side note, we also prove that the above-mentioned linear-time algorithm is optimal up to a constant factor in the worst-case sense. This improves upon previous results which had a logarithmic gap \cite{mst-laplace}. Similarly, we prove that the exponential mechanism is for MST worst-case optimal up to a constant factor for both the $\ell_1$ and the $\ell_\infty$ neighbor relations.

\subsection{Technical Overview} \label{sec:techniques}

In this section, we briefly discuss the intuition and techniques behind our results. Here, we focus mostly on the $\ell_1$ neighbor relation and universal near-optimality. The lower bound arguments for $\ell_\infty$ and the worst-case optimality are similar.
The comparison between ours and previous results is summarized in \cref{tab:results}.

Our lower and upper bounds rely on the properties of the set
$\trees(G)$ of all spanning trees of a given (unweighted) graph $G$. We define a Hamming-like metric $\ham$ in this space, defined as $\ham(T_1, T_2) = |T_1 \setminus T_2|$. It turns out that the diameter $D$ of $\trees(G)$ with respect to $\ham$ is a natural parameter that determines the ``hardness'' of $G$. Namely, we show an $\Omega(D/\eps)$ lower bound and an $\O(D \log n / \eps)$ upper bound on the expected error of the optimal $\eps$-differentially private algorithm. 

\paragraph{Upper bound.}
In \cref{sec:postprocessing}, we use the diameter $D$ to obtain a sharper analysis of the Laplace mechanism of \citet{mst-laplace}. The mechanism is very simple: add Laplacian noise to every edge, then return the MST with respect to the noisy weights. Its standard analysis uses the fact that with high probability, the noise on every edge is $\O(\log n / \eps)$, and thus the total error of the spanning tree returned is $\O(n\log n / \eps)$. Our improvement follows from the fact that the returned spanning tree differs from the MST in at most $2D$ edges, and thus the total error accumulated is actually only $\O(D \log n / \eps)$.

\paragraph{Lower bound.}
The most technically interesting part of this paper is our lower bound in \cref{sec:lower-bound}. We have a fixed underlying graph $G$ and for any $\eps$-differentially private mechanism, we want to find a weight assignment $\vew^-$ on which the error will be large. Our lower bound builds on a general packing-based lower bound for differentially private algorithms, which we briefly paraphrase in the language of MSTs and universal optimality: Given a set $\W$ of weight vectors and a parameter $x \ge 0$, define for each $\vew \in \W$ the set $\light_\vew$ of spanning trees that are \emph{$x$-light} for this $\vew$, i.e., that are heavier than the MST by at most $x$. Moreover, assume that the sets $\light_\vew$ are pairwise disjoint, that is, each spanning tree is $x$-light with respect to at most one weight vector in $\W$. Then the expected error of any $\eps$-differentially private algorithm is $\Omega(x \log |\W| / (r\eps))$ on at least one weight assignment $\vew^- \in \W$, where $r$ is the diameter of $\W$ in the metric induced by the neighbor relation $\sim$. 

Intuitively speaking, we have a set $\W$ of weight assignments and for each of them, we consider a ball of outputs that are ``good'' in the sense that their error is ``small'' (with respect to $x$) for this input. Our goal is to find a large collection $\W$ of weights that are ``close'' (i.e., $r$ is small), but which induce disjoint balls that are ``wide'' (i.e., we can set $x$ to be large).

\paragraph{Reduction to finding many dissimilar spanning trees.}
The problem of finding $\W$ that maximizes the expression $x \log|\W|/(r\eps)$ is somewhat complicated by the fact that we have a trade-off between $|\W|$, $x$ and $r$. In order to solve this problem, we first show how to reduce it to a combinatorial problem of finding a large set of dissimilar spanning trees.

The main idea of the reduction is as follows: assume we have $S \subseteq \trees(G)$ such that every two $T_1, T_2$ in $S$ have $\ham(T_1, T_2) > d$. We construct $\W$ by, for each $T \in S$, creating a weight vector $\vew_T$ by defining $\vew_T(e) = 0$ if $e \in T$ and $1$ otherwise. Now the crucial observation is that $\vew_T(\cdot) = \ham(T, \cdot)$. That is, for any $T'$, the weight of $T'$ under the weight vector $\vew_T$ is exactly the Hamming distance between $T$ and $T'$.
One can then verify that for every $T' \in \trees(G)$, there is at most one weight $\vew_T \in \W$ under which its weight is at most $d / 2$, as otherwise we would, by the triangle inequality, for some $T_1, T_2 \in S$ have that $\ham(T_1, T_2) \le \ham(T_1, T') + \ham(T_2, T') \le d$, which would be a contradiction with the definition of $d$. Therefore, we can set $x$ as large as $d / 2$ while still ensuring the disjointness property of all $\light_\vew$ required by the original lower bound. 

An upper bound of $r \leq 2D$ can be argued as follows: We have that any two spanning trees differ in $\leq 2D$ edges. Combined with having zero-one weights, any two spanning trees' weight vectors $\vew_T$ thus also have $\ell_1$-distance at most $2D$, and thus $r \le 2D$.

Using these ideas, we have reduced the original problem to the problem of finding a large set $S$ of spanning trees that is sparse in the sense that no two spanning trees in $S$ are similar. More specifically, we have reduced the original problem to the problem of finding a set $S$ of spanning trees which maximizes $\log |S| / d$.

\paragraph{How to find many dissimilar spanning trees?}

Finally, the problem of finding a large, yet sparse $S$, is reducible to a standard problem: We show that there are always at least $2^D$ spanning trees, and the problem of finding a sparse subset among them can be reduced to finding a binary code of length $D$ and minimum Hamming distance $\Omega(D)$ that has $2^{\Omega(D)}$ many codewords. Such a code exists by the Gilbert–Varshamov bound \cite{gilbert-varshamov1,gilbert-varshamov2}. The $S$ that we get from this reduction then allows us to prove an $\Omega(D / \eps)$ lower bound on the expected error, which nearly matches the $\O(D \log n / \eps)$ upper bound mentioned above.

\subsection{Related Work}
\label{sec:related_work}

\aboverulesep=0pt
\belowrulesep=0pt
\renewcommand\arraystretch{1.3}
\begin{table}[t]
\centering
\adjustbox{max width=\columnwidth}{
\begin{tabular}{ c || c | c || c || c | c | c}
	\multirow{2}{*}{\textbf{neighborhood}} & \multicolumn{2}{c||}{\textbf{on a fixed graph topology}} & \textbf{worst-case} & \multicolumn{3}{c}{\textbf{previous work (worst-case)}}\\
	\cline{2-7}
	& lower bound & upper bound & lower and upper bound & lower bound & upper bound & reference  \\ 
 \midrule[.8pt]
	$\ell_1$ & $\Omega(D/\eps)$ & $\O(D \log n/\eps)$ & $\Theta(n\log n/\eps)$ & $\Omega(n/\eps)$  & $\O(n \log n/\eps)$ & \cite{mst-laplace}\\
	$\ell_\infty$ & $\Omega(D^2/\eps)$ & $\O(D^2 \log n/\eps)$ & $\Theta(n^2\log /\eps)$ & ---  & $\O(n^2 \log n/\eps)$ & \cite{mst-pamst}\\
\end{tabular}
}
\caption{Summary of our results. Recall that $D$ is the diameter of the space $\mathcal{T}(G)$ of spanning trees of $G$, as defined in \Cref{sec:preliminaries-spanning-trees}.}
\label{tab:results}
\end{table}

\paragraph{Differentially private minimum spanning trees}

The problem of privately releasing the MST in the ``public graph, private weights'' setting was first studied by \citet{mst-laplace}. The author shows a worst-case lower bound of $\Omega(n / \eps)$ on the expected error for the $\ell_1$ neighbor relation, proposes a simple mechanism based on adding Laplacian noise to all weights and calculating the MST with respect to the noisy weights, and proves that it is worst-case optimal up to an $\O(\log n)$ factor. In \cref{sec:postprocessing,sec:lower-bound}, we show that this mechanism is in fact worst-case optimal up to a constant factor and also universally optimal up to an $\O(\log n)$ factor.

\citet{mst-pamst} gives a mechanism for the $\ell_\infty$ neighbor relation by running a noisy version of the Jarník-Prim algorithm that, in each step, selects the edge to be included in the tree by running the exponential mechanism. It has an $\O(n^2 \log n / \eps)$ expected error, which, by our results from \cref{sec:lower-bound}, is worst-case optimal. However, the analysis does not seem to be easily modifiable to show universal optimality. We remark that the author claims an expected error of $\O(n^2 \log n / (m\eps))$, but this is only true if one normalizes all weights by $1/m$.

We also note that releasing the \textit{weight} of the minimum spanning tree is easy. One may easily show that the global sensitivity is $1$ (with $\ell_1$) or $n - 1$ (with $\ell_\infty$),
allowing us to simply find the MST and release its weight using the Laplace mechanism.
In their work on smooth sensitivity, \citet{smooth-sensitivity} consider the problem of privately releasing the weight of the MST in a slightly different setting where the weights are bounded and neighboring datasets differ by changing the weight of one edge.

\paragraph{Other differential privacy notions on graphs}
The notion of privacy used in this work was introduced by \citet{mst-laplace} and was since used widely \cite{mst-pamst,pinot2018graph,chen2022all,fan2022private,brunet2016edge}. To paraphrase a real-world motivation, the graph may represent a (publicly known) road network, with edge weights corresponding to a measure of congestion. User's current location is private information which contributes to the congestion of an edge and should be protected.
Other common privacy notions include edge differential privacy and node differential privacy \cite{hay2009accurate}, where the graph itself is unweighted and private, and two graphs $G$ and $G'$ are neighboring if one can be obtained from the other by deleting an edge (for edge privacy) or a node and all its adjacent edges (for node privacy).

\paragraph{Instance-optimality in differential privacy}
The notion of universal optimality is closely linked to that of instance optimality, which states that our algorithm is ``as good as any mechanism could be'' on every single instance. Universal optimality can then be seen as a combination of instance-optimality w.r.t.~the underlying graph and worst-case optimality w.r.t.~the edge weights.

In the last few years, several instance-optimality results have been proven in differential privacy \cite{huang2021instance,dong2022nearly,blasiok2019towards,asi2020instance}. We highlight here \citet{huang2021instance} which give an instance-optimal differentially private mechanism for releasing the mean. We also highlight \citet{asi2020instance} who introduce the inverse sensitivity mechanism, and give instance-optimality results for mean estimation, performing linear regression, and the principal component analysis.

It should be noted that there are subtleties in the precise definitions of instance optimality that these papers use. The reason is that under the definition of instance-optimality that is commonly used in other areas, often no instance-optimal mechanism exists in differential privacy for trivial reasons.
Therefore, the precise definitions used in the mentioned papers differ somewhat.

\paragraph{Issues with nomenclature}
The name ``universal optimality'' has unfortunately been used to mean different things in different contexts. Throughout this paper, we use universal optimality in the sense in which it is commonly used in distributed algorithms \cite{haeupler2021universally}. It should be noted that there have been several works in differential privacy that use the name ``universal optimality'' to denote completely unrelated concepts \cite{ghosh2009universally,fernandes2022universal}. Namely, these papers use the name universal optimality in a Bayesian setting, where universal optimality states that a given mechanism is optimal no matter the prior. This is a completely unrelated notion to what we consider in this paper.

\paragraph{Universal optimality}
The notion of universal optimality started in distributed algorithms where multiple classic combinatorial problems are now known to have universally optimal algorithms in various settings \cite{haeupler2021universally,zuzic2022universally,haeupler2022hop,rozhovn2022undirected}. We highlight here the paper by \citet{haeupler2021universally} which gives a universally optimal algorithm in the supported CONGEST model for the minimum spanning tree; the techniques used in that paper are different from those that we use, despite both papers considering the MST problem.

Recently, it was shown that Dijkstra's algorithm is universally optimal for a version of the single-source shortest paths problem \cite{haeupler2023universal} in the standard Word-RAM model. To the best of our knowledge, that paper was the first universal optimality result in a non-distributed setting. This makes this paper the second such result.

\subsection{Future work}
We believe that our technique can be applied to other graph problems. Specifically, the framework that we use for our lower bound seems to be quite general. We believe that it is likely that the same techniques could be used for releasing shortest-path trees. With some additional tweaks, we believe our techniques could be useful for problems such as maximum-weight matching or minimum-weight perfect matching. Extending our results to approximate differential privacy would also be of great interest.

\section{Preliminaries}

In this paper, we consider all graphs to be simple, undirected and connected. We consider the setting introduced by \citet{mst-laplace}, where the \emph{unweighted} graph $G = (V, E)$ is public and fixed, and the only private information are the edge weights $\vew \in \R^E$. We also assume that $G$ has at least two different spanning trees, i.e., it is not itself a tree.
We denote by $\trees(G)$ the set of all spanning trees of $G$, and identify each $T \in \trees(G)$ with the set of its edges. It is folklore that for all connected graphs, $1 \le |\trees(G)| \le n^{n-2}$, with the maximum attained when $G$ is a clique. We write $\vew(T)$ as a shortcut for $\sum_{e \in T} \vew(e)$, and write $T^*_\vew$ to denote the minimum spanning tree under weights $\vew$. If $\vew$ is clear from context, we write just $T^*$. For a spanning tree $T$, its \emph{error} is defined as $\vew(T) - \vew(T^*)$.

\subsection{Differential Privacy}
We define two notions of adjacency for weight vectors: $\simone$, defined so that $\vew \simone \vew'$ if and only if $\|\vew-\vew'\|_1 \le 1$, and $\siminf$, defined so that $\vew \siminf \vew'$ if and only if $\|\vew-\vew'\|_\infty \le 1$.

A \emph{mechanism for MST} is any randomized algorithm that, for a fixed and public unweighted graph $G$, takes as input a weight vector $\vew$ and outputs a spanning tree of $G$,
specified by the list of its edges.\footnote{Formally, we have inifinitely many mechanisms, each for one graph $G$. What we will do instead is to pretend that there is a single mechanism that accepts $G$ as an additional parameter.} We are interested in minimizing the \emph{expected error} of $\algo$, defined as $\E_{T \sim \algo(G, \vew)}[\vew(T)] - \vew(T^*)$. Furthermore, $\algo$ is
\emph{$\eps$-differentially private} with respect to $\sim$ (which is either $\simone$
or $\siminf$), if, for every $\vew \sim \vew'$ and every $T \in \trees(G)$,
\[
	\Pr\left[\algo(G, \vew) = T\right] \le e^\eps \cdot \Pr\left[\algo(G, \vew') = T\right],
\]
where $\eps > 0$ is a parameter that controls the tradeoff between privacy and
the expected error of $\algo$.

\subsection{Spanning Trees}
\label{sec:preliminaries-spanning-trees}

We define the following metric on the space of spanning trees: $\ham(T_1, T_2) \coloneq |T_1 \setminus T_2|$.\footnote{The reader is invited to verify that $\ham$ is indeed a metric.} We call $\ham$ the \emph{Hamming metric} or \emph{Hamming distance} between $T_1$ and $T_2$. Note that $|T_1 \setminus T_2| = |T_2 \setminus T_1|$ and we could have used either in the previous definition. We then define $\diam(G) \coloneqq \max_{T_1, T_2 \in \trees(G)} \ham(T_1, T_2)$ as the diameter of the space of spanning trees of $G$ with respect to $\ham$. Trivially, $\diam(G) \le n - 1$, but it can be much smaller: for example, we have $\diam(G) = 1$ when $G$ is a cycle, and generally, $\diam(G) \le k$ if $G$ has at most $n - 1 + k$ edges. In \cref{sec:spanning-trees}, we provide an exponential lower and upper bound on the relationship between $\diam(G)$ and $|\trees(G)|$.

In the lower bounds, we will make use of special zero-one weight vectors, derived from spanning trees:

\begin{definition}
	\label{def:oneweights}
	For a spanning tree $T \in \trees(G)$, define the weights $\oneweights_T \in R^E$ as
	\[
		\oneweights_T(e) = 
		\begin{dcases}
			0 & \text{if }e \in T,\\
			1 & \text{otherwise}.
		\end{dcases}
	\]
\end{definition}
\begin{fact}
	\label{fact:oneweights-properties}
	For two trees $T_1, T_2 \in \trees(G)$, it holds $\oneweights_{T_1}(T_2) = |T_2 \setminus T_1| = \ham(T_1, T_2) = |T_1 \setminus T_2| = \oneweights_{T_2}(T_1)$. In particular, $\oneweights_T(T) = 0$. It also holds that $\|\oneweights_{T_1} - \oneweights_{T_2}\|_1 = 2\ham(T_1, T_2)$ and $\|\oneweights_{T_1} - \oneweights_{T_2}\|_\infty \le 1$ (with equality if $T_1 \ne T_2$).
\end{fact}

\subsection{Universal Optimality}

Universal optimality is a notion that, intuitively speaking, says that an algorithm is optimal for any fixed graph topology. It is most commonly used with respect to the time complexity, but here, we instead define it in terms of an expected error of an $\eps$-differentially private mechanism.

\begin{definition}
	\label{def:universal-optimality}
Let $\mathcal P$ be any optimization problem on weighted graphs where for every input $(G, \vew) \in \X$,
the goal is to produce an output $y \in \Y$ minimizing some scoring function
$\mu(G, \vew, y)$. Define $\mu^*(G, \vew) = \min_{y \in \Y} \mu(x, y)$. For a fixed graph $G$ and an algorithm $\algo$, define worst-case expected error of $\algo$ on $G$ as
\[
	R(\algo, G) \coloneqq \max \Big\{\, \E \left[ \mu(G, \vew, \algo(G, \vew)) \right] - \mu^*(G, \vew) \;\Big|\; \vew : (G, \vew) \in \X\,\Big\}.
\]

We say that an $\eps$-differentially private mechanism $\algo : \X \to \Y$ is \emph{universally optimal} for $\mathcal P$, if there exists a constant $c > 0$ such that for any unweighted graph $G$ and any other $\eps$-differentially private mechanism $\algo^*$, we have
\[
	R(\algo, G) \le c \cdot R(\algo^*, G).
\]
We instead say that $\algo$ is universally optimal up to factor $f(G)$ if
\[
	R(\algo, G) \le c \cdot f(G) \cdot R(\algo^*, G).
\]
\end{definition}

\subsection{Exponential Mechanism}

For completeness, we restate the guarantees of the exponential mechanism.

\begin{fact}[Guarantees of the exponential mechanism \cite{exp-mechanism,exp-mechanism-expectation}]\label{lem:exp-mechanism-guarantees}
	Let $\mu : \X \times \Y \to \R$ be a function, and let $\sim$ be a neighbor relation on $\X$. The \emph{exponential mechanism} $\algo : \X \to \Y$ that, given $x \in \X$, samples $y \in \Y$ with probability proportional to $\exp(-\frac{\eps}{2\Delta}\cdot\mu(x, y))$, is $\eps$-differentially private and satisfies:
	\[
		\E_{y \sim \algo(x)}[\mu(x, y)] \le \mu(x, y^*) + \frac{2\Delta \log |\Y|}{\eps},
	\]
	where $y^*$ is the minimizer of $\mu(x, \cdot)$ and $\Delta$ is the \emph{global sensitivity} of $\mu$, defined as
	\[
		\Delta = \sup_{x \sim x' \in \X} \max_{y \in \Y} \left|\mu(x, y) - \mu(x', y)\right|.
	\]
\end{fact}

\section{Simple Mechanism for Privately Releasing the MST}
\label{sec:postprocessing}

In this section, we show a tighter analysis of a very simple mechanism for privately releasing an MST, which originally appeared in \citet{mst-laplace}. Whereas the original analysis proves that the expected error is at most $\O(n \log n/\eps)$ for the $\simone$ neighbor relation, we prove a tighter bound of $\O(D \log n/\eps)$ for $D = \diam(G)$ (note that $D \le n - 1$). We prove in \cref{sec:lower-bound} a lower-bound of $\Omega(D/\eps)$, showing that the algorithm is universally near-optimal. On the other hand, in \cref{cl:postprocessing-not-optimal-for-inf} we prove that for $\siminf$, the algorithm is neither worst-case nor universally optimal.

Our goal is to prove the following corollary. It follows from \cref{cor:ub-laplace} (which states the upper bound), \cref{thm:lb-main} (which states the lower bound) and the fact that the time complexity of \cref{alg:laplace} is dominated by the runtime of an MST algorithm. 

\begin{corollary}
	\label{cor:laplace-optimal}
	\cref{alg:laplace} is $\eps$-differentially private under the $\ell_1$ neighbor relation $\sim_1$ for any $\eps = \O(1)$. It is universally optimal up to an $\O(\log n)$ factor for releasing the MST under the $\ell_1$ neighbor relation. It runs in near-linear time.
\end{corollary}

\begin{algorithm}[MST via postprocessing]
\label{alg:laplace}
	Given the weight vector $\vew$, let $b = 1/\eps$ for $\simone$ and $b = m/\eps$ for $\siminf$. Release the noisy weights $\hat\vew$ obtained by, for each edge, independently sampling $\hat\vew(e) \coloneqq \vew(e) + \Lap(b)$, where $\Lap(b)$ is the Laplacian distribution with mean 0 and scale $b$. Calculate the MST of $(G, \hat \vew)$ and return it.
\end{algorithm}

\begin{theorem}
	\label{thm:ub-laplace}
	\cref{alg:laplace}, denoted here as $\algo$, is $\eps$-differentially private for both $\simone$ and $\siminf$. For $\simone$, and for any $\gamma > 0$, it holds that
	\[
		\Pr_{T \sim \algo(G, \vew)}\left[\,\vew(T) \le \vew(T^*) + \frac{4D \log(n/\gamma)}{\eps}\,\right] \ge 1 - \gamma,
	\]
	where $T^*$ is the MST of $(G, \vew)$ and $D = \diam(G)$.
\end{theorem}

The idea of our proof is analogous to the original proof by \citet{mst-laplace}, except that in the last step, instead of bounding the error incurred by all edges of $T$ and $T^*$, we bound the (smaller) error incurred by the $2D$ edges in which $T$ and $T^*$ differ.

\begin{proof}
	By the guarantees of the Laplace mechanism \cite{dwork2006calibrating}, we can privately release the noised weight vector, where for $\siminf$, we additionally use that $\|\vew - \vew'\|_1 \le m$ for $\vew \siminf \vew'$. The privacy of \cref{alg:laplace} then follows from the privacy of postprocessing.

	Fix $\gamma > 0$. For each edge $e$, it holds that $\Pr[\,|\vew(e) - \hat\vew(e)| \le \log (m/\gamma)/\eps\,] = 1 - \gamma / m$ by the
	definition of $\Lap(\cdot)$. By the union bound, the probability that $|\vew(e) - \hat\vew(e)| \le \log(m/\gamma)/\eps$ holds for
	all edges is at least $1 - \gamma$. Let us condition on this event.

	Let $T$ be the spanning tree returned by \cref{alg:laplace}, i.e., the MST
	of $(G, \hat\vew)$, and let $T^*$ be the MST of $(G, \vew)$. We can write
	\begin{align*}
		\vew(T) - \vew(T^*)
		&=
		\sum_{e \in T \setminus T^*} \vew(e) - \sum_{e\in T^* \setminus T} \vew(e)
		\\&
		\le
		\sum_{e \in T \setminus T^*} \hat\vew(e) - \sum_{e \in T^* \setminus T} \hat\vew(e)
		 + 2D\log(m/\gamma)/\eps
		\\&
		=
		\hat\vew(T) - \hat\vew(T^*)
		 + 2D\log(m/\gamma)/\eps
		\\&
		\le
		2D\log(m/\gamma)/\eps.
	\end{align*}
	We used, respectively, the fact that edges in $T \cap T^*$ do not contribute to the error, the fact that $|T \setminus T^*| = |T^* \setminus T| \le D$, rewriting, and the fact that $T$ is an MST under $\hat\vew$. Noting that $\log(m/\gamma) \le \log(n^2/\gamma) = 2\log(n/\gamma)$ finishes the proof.
\end{proof}

\begin{corollary}
	\label{cor:ub-laplace}
	Let $\algo$ denote \cref{alg:laplace} in the setting with $\simone$. It holds that
	\[
		\E_{T \sim \algo(G, \vew)}\left[\vew(T)\right] - \vew(T^*) \le \frac{4D(\log n+1)}{\eps},
	\]
	where $T^*$ is the MST of $(G, \vew)$ and $D = \diam(G)$.
\end{corollary}
\begin{proof}
	By \cref{thm:ub-laplace}, we have
	\[
		\Pr_{T \sim \algo(G, \vew)}\left[\,\vew(T) - \vew(T^*) \le \frac{4D(\log n + \log (1/\gamma))}{\eps}\,\right] \ge 1 - \gamma,
	\]
	The claim then follows from the following lemma, with $a = 4D / \eps$ and $b = \log n$.
\end{proof}

\begin{lemma}
	Let $X$ be a random variable and let us have $a > 0$, $b \in \R$, such that for every $\gamma \in (0, 1]$, we have $\Pr[X \le a(\log(1/\gamma) + b)] \ge 1 - \gamma$. Then $\E[X] \le a(b + 1)$.
\end{lemma}
\begin{proof}
	Equivalently, we can write
	\[
		\Pr\left[\,\frac Xa - b \le -\log \gamma\,\right] \ge 1 - \gamma.
	\]
	Now let $x \in \R$ be such that $\gamma = \exp(-x)$. We have
	\[
		\Pr\left[\,\frac Xa - b \le x\,\right] \ge 1 - \exp(-x).
	\]
	Denote $Y = \frac Xa - b$ and let $Z \sim \text{Exp}(1)$, where $\text{Exp}(\lambda)$ is the exponential distribution. For any $x \in \R$, we have $\Pr[Y \le x] \ge 1 - \exp(-x) = \Pr[Z \le x]$, i.e., $Z$ stochastically dominates $Y$, and thus $\E[Y] \le \E[Z] = 1$. Thus, by linearity of expectation, $\E[X] = a(\E[Y] + b) \le a(b + 1)$, as needed.
\end{proof}

\section{Lower Bounds for MST}
\label{sec:lower-bound}

In this section, we focus on proving the lower bound. For $\simone$, this is a lower bound that nearly matches the performance of \Cref{alg:laplace}, for $\siminf$ it nearly matches the performance of the mechanism from \Cref{sec:upper-bound}.

In both cases, there is a $\log n$ gap between the lower and upper bounds. This gap is inherent to our approach, but for some graphs, and in particular, for $G = K_n$, we can make it disappear and get worst-case lower bounds of $\Omega(n\log n/\eps)$ (for $\simone$) and $\Omega(n^2\log n/\eps)$ (for $\siminf$) that match our worst-case upper bounds from \cref{sec:upper-bound} up to a constant factor, as well as that of \citet{mst-laplace} in the case of $\simone$. Overall, we prove the following theorem:
\begin{theorem}
	\label{thm:lb-main}
	Fix an unweighted graph $G$ and let $\algo$ be mechanism for MST on $G$. If $\algo$ is $\eps$-differentially private with respect to $\simone$, then there exist weights $\vew$ such that
	\[
		\E_{T\sim \algo(G,\vew)} [\vew(T)] = \vew(T^*) + \Omega\left(\frac{D}{\eps}\right) - \O(1),
	\]
	where $D = \diam(G)$.
	Moreover, there exists at least one $G$ and $\vew$ where
	\[
		\E_{T\sim \algo(G,\vew)} [\vew(T)]
		=
		\vew(T^*) + \Omega\left(\frac{n \log n}{\eps}\right) - \O(1)
		.
	\]

	If $\algo$ is $\eps$-differentially private with respect to $\siminf$ instead, then, for every fixed $G$, there exist weights $\vew$ such that
	\[
		\E_{T\sim \algo(G,\vew)} [\vew(T)] = \vew(T^*) + \Omega\left(\frac{D^2}{\eps}\right) - \O(D).
	\]
	Moreover, there exists at least one $G$ and $\vew$ where
	\[
		\E_{T\sim \algo(G,\vew)} [\vew(T)]
		=
		\vew(T^*) + \Omega\left(\frac{n^2 \log n}{\eps}\right) - \O(n)
		.
	\]

\end{theorem}

The rough outline of the proof is as follows: in \cref{lem:lb-packing,lem:lb-packing-mst}, we state (and restate in a language of MSTs) a general packing-based lower bound that applies whenever we can find many inputs (in our case, weight vectors) close in the metric induced by the neighbor relation such that, for each input, the set of all outputs (in our case, spanning trees) that give small error with respect to this input, is disjoint with all the other sets. In \cref{lem:many-trees-imply-lb}, we prove that instead of searching for a good set of weight vectors, we can search for a large set of spanning trees in which every two spanning trees differ in many edges.
\cref{lem:many-trees-general,lem:many-trees-clique} prove that such a large set always exists, with the latter being a special case providing stronger guarantees when the graph is a clique. We postpone their proofs to \cref{sec:spanning-trees}, where we also build the needed theory.

At the core of our lower bound is the following lemma, based on a general packing argument.

\begin{lemma}[\citet{dp-complexity}, Theorem 5.13]
	\label{lem:lb-packing}
	Let $\mathcal C \subseteq \mathcal X$ be a collection of datasets all at (neighbor-relation-induced) distance at most $r$ from some fixed dataset $x_0 \in \mathcal X$, and let $\{\light_x\}_{x \in \mathcal C}$ be a collection of \emph{disjoint} subsets of $\mathcal Y$. If there is an $\eps$-differentially private mechanism $\mathcal M : \mathcal X \to \mathcal Y$ such that $\Pr[\mathcal M(x) \in \light_x] \ge p$ for every $x \in \mathcal C$, then
	\[
		p \le \frac{e^{r\eps}}{|\mathcal C|}.
	\]
\end{lemma}

The following rephrases \cref{lem:lb-packing} in the language of our problem. It also picks a specific way of constructing the sets $\mathcal L_\vew$, namely, by including precisely those outputs that have a ``small enough'' error with respect to $\vew$.

\begin{lemma}[Adaptation of \cref{lem:lb-packing} for MST]
	\label{lem:lb-packing-mst}
	Fix an unweighted graph $G = (V, E)$ and let $\sim$ be an arbitrary neighbor relation on $\R^E$. Given $\W \subseteq \R^E$ a collection of weight vectors, all at ($\sim$-induced) distance at most $r \in \N$ from some fixed weight vector $\vew_0 \in \R^E$, and a parameter $x \ge 0$, denote for each $\vew \in \W$ the MST of $(G, \vew)$ by $T^*_\vew$, and define
	\[
		\light_\vew \coloneqq \{\, T \in \trees(G) \mid \vew(T) \le \vew(T^*_\vew) + x \,\}
	\]
	the set of spanning trees that are ``light'' under $\vew$.
	Assume that $\W$ and $x$ are such that all $\light_\vew$ are pairwise disjoint. Then for any $\eps$-differentially private (with respect to $\sim$) mechanism $\mathcal M : \R^E \to \trees(G)$, there exist weights $\vew \in \W$, such that
	\[
		\Pr_{T \sim \mathcal M(\vew)}[\,\vew(T) \le \vew(T^*_\vew) + x\,]
		\le
		\frac{e^{r\eps}}{|\W|}.
	\]
\end{lemma}

%

It turns out that there is a very natural class of sets $\W$ that give a reasonably good lower bound. Namely, the $\W$ we use in our lower bound will be of the form $\W \subseteq \{0, \alpha\}^E$ for some choice of $\alpha$. This is formalized in the following statement. It says that to create our $\W$, we can start with any $S \subseteq \trees(G)$, and then use the mapping $T \to \alpha \cdot \oneweights_T$ (recall that $\oneweights_T$ is the indicator function of a set $E \setminus T$, see \cref{def:oneweights}). The quality of the lower bound provided by $S$ depends solely on two things: the size of $S$, and the minimum Hamming distance between distinct $T_1, T_2 \in S$. The parameter $\alpha > 0$ can then be used to control the tradeoff between $x$ and the probability that the error will be smaller than $x$ in \cref{lem:lb-packing-mst}.

\begin{lemma}
	\label{lem:many-trees-imply-lb}
	Fix an unweighted graph $G$ and denote $D = \diam(G)$. Given a set of spanning trees $S \subseteq \trees(G)$, let $d > 0$ be such that $\ham(T_1, T_2) > d$ for all $T_1, T_2 \in S$. Finally, let $\algo$ be a mechanism for releasing a minimum spanning tree of $G$. If $\algo$ is $\eps$-differentially private with respect to $\simone$, then there exist weights $\vew$ such that
	\begin{align*}
		\Pr_{T \sim \mathcal A(G, \vew)}\left[\,\vew(T) \le \vew(T^*) + \frac dD \left(\frac{\log|S|}{8\varepsilon} - \frac{1}{4}\right)\,\right]
		\le \frac{1}{\sqrt{|S|}}.
	\end{align*}
	If $\algo$ is $\eps$-differentially private with respect to $\siminf$, then there exist weights $\vew$ such that
	\begin{align*}
		\Pr_{T \sim \mathcal A(G, \vew)}\left[\,\vew(T) \le \vew(T^*) + d \left(\frac{\log|S|}{4\varepsilon} - \frac{1}{2}\right)\,\right]
		\le \frac{1}{\sqrt{|S|}}.
	\end{align*}
\end{lemma}

\begin{proof}
	Let us walk through the proof with the assumption that $\algo$ is $\eps$-differentially private with respect to $\simone$. At the very end, we will discuss the needed changes for $\siminf$.

	Fix a parameter $\alpha \in \R$ and define $\W$ as $\{\,
	\alpha\cdot \oneweights_T \mid T \in S \,\}$. Later we will specify the value of
	$\alpha$ that gives the needed bound.

	Fix $\vew_0 = \alpha\oneweights_{T_0} \in \W$ arbitrarily. In order to apply \cref{lem:lb-packing-mst}, we need to determine $r$ and $x$. For each $\vew = \alpha\oneweights_T \in \W$, we have $\|\vew - \vew_0\|_1 = 2\alpha\ham(T, T_0) \le 2\alpha D$, and thus $r = \lceil 2\alpha D\rceil$.

	Set $x = \alpha d / 2$. Then for every $\vew = \alpha\oneweights_T \in \W$, we have:
	\begin{align*}
		\light_\vew
		&=
		\{\,T' \in \trees(G) \mid \vew(T') < \vew(T) + x\,\}
		=
		\{\,T' \in \trees(G) \mid \alpha\oneweights_T(T') < \alpha d/2 \,\}
		\\
		&=
		\{\,T' \in \trees(G) \mid |T' \setminus T| < d/2 \,\}
		=
		\{\,T' \in \trees(G) \mid \ham(T, T') < d/2 \,\}.
	\end{align*}

	We can conclude that all $\light_\vew$ are disjoint: If we had $T \in \light_{\vew_1} \cap \light_{\vew_2}$ for distinct $\vew_1 = \alpha\oneweights_{T_1} \in \W$ and $\vew_2 = \alpha\oneweights_{T_2} \in \W$, then we would have $\ham(T_1, T_2) \le \ham(T_1, T') + \ham(T', T_2) < d$, which would contradict the properties of $d$, as $T_1, T_2 \in S$.

	Now we apply \cref{lem:lb-packing-mst} to get that, for every $\eps$-DP mechanism for releasing the MST of $G$ there exist weights $\vew \in \W \subseteq \R^E$, such that
	\begin{align*}
		\Pr_{T \sim \mathcal A(G, \vew)}\left[\,\vew(T) \le \vew(T^*) + \frac{\alpha d}{2}\,\right]
		&\le \frac{\exp(m\eps)}{|S|} = \frac{\exp(\lceil 2\alpha D\rceil\eps)}{|S|} \le \frac{\exp((2\alpha D + 1)\eps)}{|S|}.
	\end{align*}
	Setting $\alpha = 1/4 \cdot \log |S| / (\eps D) - 1/2D$ yields:
	\begin{align*}
		\Pr_{T \sim \mathcal A(G, \vew)}\left[\,\vew(T) \le \vew(T^*) + \frac dD \left(\frac{\log|S|}{8\varepsilon} - \frac{1}{4}\right)\,\right]
		\le \frac{1}{\sqrt{|S|}}.
	\end{align*}

	Finally, let us deal with the case when $\algo$ is $\eps$-differentially private with respect to $\siminf$ instead. The idea of the proof is identical and we only point out the differences. When calculating $r$, we now get that $\|\vew - \vew_0\|_\infty \le \alpha$, and thus $r = \lceil \alpha\rceil$. By \cref{lem:lb-packing-mst}, and as $\lceil \alpha\rceil \le \alpha + 1$, there exist weights $\vew$ such that
	\begin{align*}
		\Pr_{T \sim \mathcal A(G, \vew)}\left[\,\vew(T) \le \vew(T^*) + \frac{\alpha d}{2}\,\right]
		&\le\frac{\exp((\alpha + 1)\eps)}{|S|}.
	\end{align*}
	Setting $\alpha = 1/2 \cdot \log |S| / \eps - 1$ yields:
	\begin{align*}
		\Pr_{T \sim \mathcal A(G, \vew)}\left[\,\vew(T) \le \vew(T^*) + d\left(\frac{\log|S|}{4\varepsilon} - \frac{1}{2}\right)\,\right]
		\le \frac{1}{\sqrt{|S|}}.
	\end{align*}
\end{proof}

\Cref{lem:many-trees-imply-lb} reduces our problem into a problem of finding a set $S \subseteq \trees(G)$, which is as large as possible, and at the same time, the Hamming distance between any two spanning trees in the set is large. In \cref{sec:spanning-trees}, we prove the following two results:

\begin{lemma}
	\label{lem:many-trees-general}
	For any unweighted graph $G$ with $\diam(G) = D$, there exists a set $S \subseteq \trees(G)$ of size $2^{\Theta(D)}$ such that for any distinct $T_1, T_2 \in S$, we have $\ham(T_1, T_2) = \Theta(D)$.
\end{lemma}

\begin{lemma}
	\label{lem:many-trees-clique}
	For any unweighted clique $G$ on $n > 2$ vertices, there exists a set $S \subseteq \trees(G)$ of size $2^{\Theta(n \log n)}$ such that for any distinct $T_1, T_2 \in S$, we have $\ham(T_1, T_2) = \Theta(n)$.
\end{lemma}

Now we are ready to prove \cref{thm:lb-main}.

\begin{proof}[Proof of \cref{thm:lb-main}]
	\cref{lem:many-trees-imply-lb} says that for any mechanism $\algo$, there exist weights $\vew$ such that,
	if $\algo$ is $\eps$-differentially private with respect to $\simone$, then:
	\begin{align*}
		\Pr_{T \sim \mathcal A(G, \vew)}\Bigg[\,\vew(T) \le \vew(T^*) + \underbrace{\frac dD\left(\frac{\log|S|}{8\varepsilon} - \frac14\right)}_{x}\,\Bigg]
		\le |S|^{-1/2}
	\end{align*}
	for $D = \diam(G)$ and $d$ such that $\ham(T_1, T_2) < d$ for any distinct $T_1, T_2 \in S$.
	This means that $\Pr[\vew(T) > \vew(T^*) + x] \ge 1 - |S|^{-1/2} = \Omega(1)$, and thus, as $d / D \le 1$:
	\[
		\E[\vew(T)] \ge \vew(T^*) + \Omega(x) = \vew(T^*) + \Omega\left(\frac{d\log|S|}{D\eps}\right) - \O(1).
	\]
	Now, if $S$ is from \cref{lem:many-trees-general}, we have $\log |S| = \Theta(D)$ and $d =
	\Theta(D)$. If $G = K_n$ then instead we take $S$ from
	\cref{lem:many-trees-clique}, and $\log |S| = \Theta(n \log n)$ and $d = \Theta(D)$. Substituting twice in the equation above proves the first half of the theorem.

	If, instead, $\algo$ is $\eps$-differentially private with respect to $\siminf$, there likewise exist weights $\vew$ such that:
	\begin{align*}
		\Pr_{T \sim \mathcal A(G, \vew)}\Bigg[\,\vew(T) \le \vew(T^*) + d\left(\frac{\log|S|}{4\varepsilon} - \frac{1}{2}\right)\,\Bigg]
		\le |S|^{-1/2},
	\end{align*}
	and by a similar argument,
	\[
		\E[\vew(T)] \ge \vew(T^*) + \Omega\left(\frac{d\log|S|}{\eps}\right) - \O(d).
	\]
	Now, we can again substitute for $\log |S|$ and $d$ according to either \cref{lem:many-trees-general} (if $G$ is a general graph), or \cref{lem:many-trees-clique} (if $G$ is a clique). This proves the second half of the theorem.
\end{proof}

\section{Finding a Large Set of Dissimilar Trees}
\label{sec:spanning-trees}

In this section, our goal is to prove \cref{lem:many-trees-general,lem:many-trees-clique}, which assert the existence of a large set $S$ of spanning trees such that every two spanning trees differ in many edges. In~\cref{sec:spanning-trees-basic}, we state some common properties of spanning trees and prove that there are at least $2^D$ spanning trees (for $D = \diam(G)$). In~\cref{sec:binary-codes}, we show how to embed binary block codes into the $2^D$ spanning trees from \cref{sec:spanning-trees-basic}, which leads to our first method of constructing $S$, and proving \cref{lem:many-trees-general}; this is the lemma that we need in order to prove our universal optimality results.

In~\cref{sec:greedy-packing}, we bound the number of trees in the $d$-ball around some tree $T$, and use this in conjunction with a greedy packing argument to provide a different method of constructing $S$ that gives slightly different guarantees than the one in \Cref{sec:binary-codes}. We use this to prove \cref{lem:many-trees-clique}; this is the lemma that we need to prove our worst-case optimality results.

\subsection{Properties of Spanning Trees}
\label{sec:spanning-trees-basic}

Here we state some simple properties of spanning trees. Although (some of) these results could be considered folklore, we give proofs for completeness.

\begin{lemma}[Exchange lemma]
	\label{lem:exchange}
	Let $T_x, T_y \in \trees(G)$ be two spanning trees and let $e \in T_y
	\setminus T_x$. Then there exists $f \in T_x \setminus T_y$ such that $T' =
	T_x \cup \{e\}\setminus \{f\}$ is also a spanning tree. Furthermore, $|T_y \setminus T'| = |T_y \setminus T_x| - 1$.
\end{lemma}
\begin{proof}
	The graph $T_x \cup \{e\}$ has
	exactly one cycle $C$. $T_y$ does not have cycles, and
	$C$ must thus contain an edge $f \notin T_y$; furthermore, $f \ne e$, as $e
	\in T_y$. But then $f \in C \setminus \{e\} \subseteq T_x$ and thus $f \in
	T_x \setminus T_y$ as needed, and $T' = T_x \cup \{e\} \setminus \{f\}$ is
	again a tree. Finally, $|T_y \setminus T'| = |T_y \setminus T_x| - 1$, as
	removing $f$ had no effect on $|T_y \setminus T'|$ and adding $e$ decreased
	it by $1$.
\end{proof}

\begin{lemma}[Iterated exchange lemma]
	\label{lem:exchange++}
	Given two spanning trees $T_a, T_b \in \trees(G)$, and a set of edges $Q
	\subseteq T_b \setminus T_a$, there exists a spanning tree $T_Q$ such that
	$T_Q \setminus T_a = Q$ and $|T_b \setminus T_Q| = |T_b \setminus T_a| - |Q|$.
\end{lemma}
\begin{proof}
	Let $k = |T_b \setminus T_a|$. The claim follows by induction on $|Q|$. If $|Q| =
	0$, then we can set $T_Q \coloneqq T_a$. Otherwise, take any
	$e \in Q$, and let $Q' = Q \setminus \{e\}$. By induction,
	let $T'$ be the tree such that $T' \setminus T_a = Q'$ and $|T_b \setminus T'| = |T_b \setminus T_a| - |Q'| = k - |Q| + 1$. Now 
	invoke \cref{lem:exchange} with $T_x = T'$, $T_y = T_b$ and $e$ to obtain $f \in T' \setminus T_b$
	and a tree $T = T' \cup \{e\} \setminus \{f\}$ satisfying $|T_b \setminus T| = |T_b \setminus T'| - 1
	= k - |Q|$.
	
	We claim that $T \setminus T_a \subseteq Q$, since $T' \setminus T_a = Q'
	\subseteq Q$ by induction, and the only edge $T$ additionally contains is
	$e \in Q$. But simultaneously, by triangle inequality,
	$
		\ham(T_a, T) \ge \ham(T_a, T_b) - \ham(T, T_b) = k - (k - |Q|) = |Q|.
	$ 
	Hence, $T \setminus T_a \subseteq Q$ and $|T \setminus T_a| = |Q|$, and thus $T \setminus T_a = Q$ as needed and we can set $T_Q \coloneqq T$.
\end{proof}

\begin{corollary}
	\label{lem:exp-many-trees}
	For every graph $G$, it holds that $|\trees(G)| \ge 2^D$ for $D = \diam(G)$.
\end{corollary}
\begin{proof}
	Fix any two $T_a, T_b \in \trees(G)$ such that $|T_b \setminus T_a| = D$. By \cref{lem:exchange++}, we have that for any
	$Q \subseteq T_b \setminus T_a$, there exists a spanning tree $T_Q$ such that $T_Q
	\setminus T_a = Q$. As there are $2^D$ possible choices of $Q$, and each of them yields a unique spanning tree, we can conclude that there are at least $2^D$ different spanning trees.
\end{proof}

\subsection{Dissimilar Trees via Binary Codes}
\label{sec:binary-codes}

In this section, we show that the set $\mathcal Z$ of $2^D$ trees from \cref{lem:exp-many-trees}, behaves, in some sense, as the space $\{0, 1\}^D$. Namely, there is a correspondence between $\mathcal Z$ and $\{0, 1\}^D$, such that if $\vex, \vey \in \{0, 1\}^D$ differ in $k$ positions, then their corresponding spanning trees have Hamming distance at least $k / 2$.
In this way, we reduce the problem of finding a set $S \subseteq \mathcal Z$ of dissimilar spanning trees to the problem of finding a good binary block code.

\rh{statement jde zjednodušit = zobecnit, ale chceme to?}
\begin{lemma}
	\label{lem:intermediate-spanning-trees}
	Given two spanning trees $T_a$, $T_b$, let $Q_1, Q_2 \subseteq T_b \setminus T_a$. Let $T_{Q_1}$ and $T_{Q_2}$ be trees such that $T_{Q_1} \setminus T_b = Q_1$ and $T_{Q_2} \setminus T_b = Q_2$. Then it holds 
that $Q_1 \setminus Q_2 \subseteq T_{Q_1} \setminus T_{Q_2}$.
\end{lemma}
We note that $T_{Q_1}$ and $T_{Q_2}$ always exist thanks to \cref{lem:exchange++}.
\begin{proof}
	By rules for set subtraction, we have
	$
	Q_1 \setminus Q_2 = (T_{Q_1} \setminus T_b) \setminus (T_{Q_2} \setminus T_b) = (T_{Q_1} \setminus T_{Q_2}) \setminus T_b \subseteq T_{Q_1} \setminus T_{Q_2}.
	$
\end{proof}

We briefly recall the definition of a block code.
Note that we use the less common definition of an $(n, M, d)_2$ code where $M$ is not the message length, but the (exponentially larger) number of codewords.

\begin{definition}
	A \emph{$(n, M, d)_2$ code} is a set $\C \subseteq \{0, 1\}^n$ such that $|\C| \ge M$ and every two different vectors $\vex, \vey \in \C$ differ in at least $d$ positions.
\end{definition}

Next, we prove the reduction between finding a set of dissimilar spanning trees and finding a good block code:

\begin{lemma}
	\label{lem:code-to-trees}
	Let $\C$ be a $(D, M, d + 1)_2$ code and let $G$ be an unweighted graph such that $D = \diam(G)$. Then there exists a set $S \subseteq \trees(G)$ such that $|S| = M$ and $\ham(T_1, T_2) > d / 2$ for any two different $T_1, T_2 \in S$.
\end{lemma}
\begin{proof}
	We will construct $S$ as follows: first, we fix $T_a$ and $T_b$ such that $\ham(T_a, T_b) = D$, as in \cref{lem:exp-many-trees}. We number the edges of $T_b \setminus T_a$ in any order as $e_1, \ldots, e_D$. Now, for every $\vex \in \C$, define $Q_\vex\coloneq \{\, e_i \mid i \in \{1, \ldots, D\} \land \vex_i = 1 \,\}$, and let $T_{Q_\vex}$ be the tree obtained by invoking \cref{lem:exchange++}. Namely, $T_{Q_\vex}$ satisfies $T_{Q_\vex} \setminus T_a = Q_\vex$. Finally, we set $S = \{\, T_{Q_\vex} \mid \vex \in \C \,\}$.

	Clearly $|S| = |\C| = M$, and we need to show that $\ham(T_{Q_\vex}, T_{Q_\vey}) > d / 2$ for all distinct $\vex, \vey \in \C$. Note that the number of positions in which $\vex$ and $\vey$ differ is exactly $|Q_\vex \setminus Q_\vey| + |Q_\vey \setminus Q_\vex|$. Now we can use \cref{lem:intermediate-spanning-trees}, to write
	\[
		d + 1
		\le
		|Q_\vex \setminus Q_\vey| + |Q_\vey \setminus Q_\vex|
		\le
		|T_{Q_\vex} \setminus T_{Q_\vey}| + |T_{Q_\vey} \setminus T_{Q_\vex}|
		=
		2 \ham(T_{Q_\vex}, T_{Q_\vey}),
	\]
	and thus $\ham(T_{Q_\vex}, T_{Q_\vey}) > d / 2$, as needed.
\end{proof}

Finally, we show that there always exists a good block code:

\begin{lemma}
	\label{lem:good-block-code}
	For every $n$, there exists an $(n, 2^{\lfloor n/3\rfloor}, \lfloor n/6\rfloor + 1)_2$ code.
\end{lemma}
\begin{proof}
	First assume that $n$ is divisible by $6$.
	By the Gilbert–Varshamov bound \cite{gilbert-varshamov1,gilbert-varshamov2}, there exists a $(n, K, n/6 + 1)_2$ code for some $K$ that satisfies
	\[
		K \ge \frac{2^n}{\sum_{i=0}^{n/6}\binom{n}{i}}.
	\]

	Using the bound on the sum of binomial coefficients (see e.g.~\citet[p.~427]{flum2006parameterized}), we obtain $\sum_{i=0}^{n/6} \le 2^{n \cdot H(1/6)}$ with $H(p)$ being the binary entropy function $H(p) = -p\log p - (1-p) \log (1-p)$. One can verify that $H(1/6) \le 2/3$ and thus $K \ge 2^n / 2^{n \cdot 2/3} = 2^{n/3}$
	as needed.

	Finally, if $n$ is not divisible by $6$, we can apply the lemma with $n' = 6\lfloor n/6\rfloor$ and then pad every codeword of the resulting code with $n - n'$ zeros.
\end{proof}

Combination of these results proves \cref{lem:many-trees-general}:

\begin{proof}[Proof of \cref{lem:many-trees-general}]
	Combining \cref{lem:code-to-trees,lem:good-block-code} and setting $n = D$ in the latter immediately yields $S \subseteq \trees(G)$ with $|S| = 2^{\Theta(D)}$ and $d = \Theta(D)$.
\end{proof}

\subsection{Dissimilar Trees via Greedy Packing}
\label{sec:greedy-packing}

In this section, we provide a different approach for finding a large set $S$ of
dissimilar spanning trees. It works by producing a crude upper bound $U$ on the number of trees in the $d$-ball around a tree $T$, and then using a greedy packing argument to show that we can always find $S$ of size $|S| \ge |\trees(G)| / U$.

\begin{lemma}[Volume of a $d$-ball around a spanning tree]
	\label{lem:ball-bound}
	For a graph $G$, $T \in \trees(G)$, and $d \in \R_{>0}$, it holds that
	\[
		|\{\,T' \in \trees(G) \mid \ham(T, T') \le d\,\}| \le m^d n^d.
	\]
\end{lemma}
\begin{proof}
	Any $T'$ with $\ham(T, T') = d' \le \lfloor d\rfloor$ can be fully (and possibly
	non-uniquely) described by a list $L_+$ of $d'$ edges $e \in E \setminus T$ to be
	added to $T$ and another list $L_-$ of $d'$ edges $e \in T$ to be removed.
	Furthermore, both lists can be padded to have length exactly $\lfloor d\rfloor$ by
	repeating arbitrary entries. As there are at most $(m-n+1)^{\lfloor d\rfloor} (n-1)^{\lfloor d \rfloor} \le m^{\lfloor d \rfloor} n^{\lfloor d \rfloor}$ possible pairs $(L_+, L_-) \in (E \setminus T)^{\lfloor d \rfloor} \times T^{\lfloor d \rfloor}$ of lists of length ${\lfloor d \rfloor}$, there must be at most $m^{\lfloor d \rfloor} n^{\lfloor d \rfloor} \le m^dn^d$ possible trees $T'$.
\end{proof}

\begin{lemma}[Greedy packing]
	\label{lem:tree-set-by-ball-cutting}
	For a graph $G$, and a parameter $d > 0$, there exists a set $S \subseteq \trees(G)$ such that $\ham(T_1, T_2) > d$ for all distinct $T_1, T_2 \in S$ and furthermore,
	\[
		|S| \ge \frac{|\trees(G)|}{m^dn^d}.
	\]
\end{lemma}
\begin{proof}
	We will construct $S$ greedily: start with $X = \trees(G)$. As long as $X$
	is nonempty, pick arbitrary $T \in X$ and add it to $S$. Then, set $X
	\coloneqq X \setminus \{\,T' \mid \ham(T, T') < d\,\}$, and repeat. By
	\cref{lem:ball-bound}, in each step we remove at most $m^d
	n^d$ elements from a set of size $|\trees(G)|$, and therefore only stop
	after $S$ contains at least $|\trees(G)| / (m^dn^d)$ elements.
\end{proof}

Now we are ready to prove \cref{lem:many-trees-clique}:

\begin{proof}[Proof of \cref{lem:many-trees-clique}]
	We invoke \cref{lem:tree-set-by-ball-cutting} on $G$ with $d = (n-2)/6$.
	Since $G$ is a clique, we have $|\trees(G)| = n^{n - 2}$. Clearly, $d =
	\Theta(n)$ and $|S| \le |\trees(G)| = 2^{\O(n\log n)}$, and we can write
	\[
		|S|
		\ge \frac{|\trees(G)|}{m^d n^d}
		\ge \frac{|\trees(G)|}{n^{3d}}
		= \frac{|\trees(G)|}{n^{(n - 2) / 2}}
		= \frac{|\trees(G)|}{\sqrt{|\trees(G)|}} =\sqrt{|\trees(G)|} = 2^{\Omega(n \log n)}.
	\]
\end{proof}

\cref{lem:ball-bound} also gives us the following relationship between $|\trees(G)|$ and $D$:

\begin{lemma}
	\label{lem:numtrees-vs-diam}
	For a graph $G$ with $\diam(G) = D$, it holds that $|\trees(G)| \le 2^{3D\log_2 n}$.
\end{lemma}
\begin{proof}
	Take any $T \in \trees(G)$ and define $S = \{\,T' \in \trees(G) \mid \ham(T, T') \le D\,\}$. Necessarily $S = \trees(G)$ by the definition of $D$. Now we invoke \cref{lem:ball-bound} to conclude that $|\trees(G)| = |S| \le m^D n^D \le n^{3D} = 2^{3D\log_2 n}$.
\end{proof}

\section{Universal Near-Optimality via the Exponential Mechanism}
\label{sec:upper-bound}

In this section, we start by showing that \Cref{alg:laplace} is in fact not universally optimal for $\siminf$.
The main contribution of this section is then that we prove that the exponential mechanism is universally near-optimal with respect to both $\simone$ and $\siminf$. We also show that it can be implemented in polynomial time by relying on a result of \citet{mst-sampling-in-matrix-multiplication}.

Our goal is to prove the following corollary. It follows from \cref{thm:ub-exp-mechanism} (which states the upper bounds and time complexity) and \cref{thm:lb-main} (which states the lower bounds).

\begin{corollary}
	\label{cor:exp-mechanism-optimal}
 For any $\eps = \O(1)$, the exponential mechanism with loss function $\mu(\vew, T) = \vew(T)$ is universally optimal up to an $\O(\log n)$ factor for releasing the MST, in both the $\ell_1$ and the $\ell_\infty$ neighbor relations. It can be implemented in the matrix multiplication time~$\O(n^\omega)$.
\end{corollary}

We now show that \Cref{alg:laplace} is neither worst-case, nor universally optimal. Note also that when we are using the $\siminf$ neighbor relation, the noise magnitude used by the algorithm is indeed optimal in the sense that any lower noise magnitude will not lead to the weights themselves being private after adding the noise. This can be easily seen as follows: With the current amount of noise added, if each weight changes by 1, we lose up to $\eps/m$ privacy on each edge. Since the composition theorem for pure differential privacy is tight, we thus may indeed lose up to $\eps$ privacy in total. Any lower amount of noise would not give $\eps$-differential privacy.
\begin{claim}
	\label{cl:postprocessing-not-optimal-for-inf}
	Denote \cref{alg:laplace} used with the neighbor relation $\siminf$ as $\algo$. For every graph $G$, there exist weights $\vew$ such that:
	\[
		\E_{T\sim \algo(G,\vew)} [\vew(T)] = \vew(T^*) + \Omega\left(\frac{mD}{\eps}\right) - \O(1).
	\]
\end{claim}
\begin{proof}
	The claim follows immediately from the fact that $\algo$ is also $\eps/m$-differentially private with respect to $\simone$, and thus by the first part of \cref{thm:lb-main} with $\eps' = \eps/m$, such weights $\vew$ must exist.
\end{proof}

Below, we will prove that one can in fact achieve an expected error of $\O(D^2 \log n/\eps)$. This implies that \Cref{alg:laplace} is in fact neither universally, nor worst-case optimal for $\siminf$.

In the rest of this section, our goal is to prove that the exponential mechanism can be implemented in polynomial time and that it is universally optimal for releasing the MST. We start by stating a useful result on efficiently sampling spanning trees.

\begin{lemma}\label{lem:exp-mechanism-fast}
	There is an algorithm that, given an unweighted graph $G$, a weight vector $\vew$, and a parameter $\lambda \in \R$,
	samples a spanning tree of $G$ such that the probability that $T \in \trees(G)$ is
	returned is proportional to $\exp(-\lambda \vew(T))$. It runs in matrix multiplication time $\O(n^\omega)$.
\end{lemma}
\begin{proof}
	Such an algorithm is provided in \citet{mst-sampling-in-matrix-multiplication}. The original paper only deals with unweighted graphs, but it is mentioned in \citet{mst-sampling-kyng} that this approach is actually easily generalized to the weighted case.
\end{proof}

Note that generally, any exact spanning tree sampling algorithm that supports weighted graphs works for \cref{lem:exp-mechanism-fast}. There are faster algorithms available, but every faster algorithm known to us either does not sample from the exact distribution (failing or sampling from a different distribution with some small probability $\delta$), or does not support weighted graphs out of the box.

The following lemma allows us to analyze the exponential mechanism more tightly by changing the loss function so that the outcome probabilities do not change, but the global sensitivity decreases.

\begin{lemma}
	\label{lem:equivalent-utilities}
	Let $\mu, \mu' : \X \times \Y \to \R$ be two loss functions related in
	the following way: for each $x \in \X$, there exists $c_x \in \R$ such that
	for all $y \in \Y$, $\mu'(x, y) = \mu(x, y) + c_x$.

	Given $\lambda \in \R$, let $\algo$ and $\algo'$ be instantiations of the exponential mechanism that, given $x \in \X$, sample $y \in \Y$ with probability proportional to $\exp(-\lambda\mu(x, y))$ and $\exp(-\lambda\mu'(x, y))$, respectively. Then $\algo$ and $\algo'$ are equivalent, that is, for each $x$, they return the same distribution on $\Y$.
\end{lemma}
\begin{proof}
	$\algo'(x)$ returns $y$ with probability
	\[
		\frac{\exp(-\lambda\mu'(x, y))}{\sum_{y'\in\Y} \exp(-\lambda\mu'(x, y'))}
		=
		\frac{\exp(-\lambda c_x) \cdot \exp(-\lambda\mu(x, y))}{\sum_{y'\in\Y} \exp(-\lambda c_x) \cdot \exp(-\lambda\mu(x, y'))}
		=
		\frac{\exp(-\lambda\mu(x, y))}{\sum_{y'\in\Y} \exp(-\lambda\mu(x, y'))},
	\]
	which is exactly the probability of $\algo(x)$ returning $y$.
\end{proof}

\begin{theorem}
	\label{thm:ub-exp-mechanism}
	There is an $\O(n^\omega)$-time mechanism $\algo$ for MST, $\eps$-differentially private with respect to $\simone$, such that, for every weighted graph $(G, \vew)$,
	\[
		\E_{T \sim \algo(G, \vew)} [\vew(T)] \le \vew(T^*) + \frac{2\log|\trees(G)|}{\eps}
		= \vew(T^*) + \O\left(\frac{D\log n}{\eps}\right),
	\]
	where $T^*$ is the MST of $(G, \vew)$ and $D = \diam(G)$.
	Furthermore, there is an $\O(n^\omega)$-time mechanism $\algo$ for MST, $\eps$-differentially private with respect to $\siminf$, such that, for every weighted graph $(G, \vew)$,
	\[
		\E_{T \sim \algo(G, \vew)} [\vew(T)] \le \vew(T^*) + \frac{4D\log|\trees(G)|}{\eps}
		= \vew(T^*) + \O\left(\frac{D^2\log n}{\eps}\right).
	\]
\end{theorem}

\begin{proof}
	The asymptotic bounds on the error follow from the exact ones by \cref{lem:numtrees-vs-diam}\rh{je tahle věta ok?}, thus we will only focus on the exact inequalities. Let us assume the $\siminf$ case first; we will deal with the $\simone$ case at the end of the proof.

	$\algo$ will be an instantiation of the exponential
	mechanism from \cref{lem:exp-mechanism-fast}, with $\mu(\vew, T) \coloneq \vew(T)$, and with $\lambda$ determined later.

	We will use~\cref{lem:equivalent-utilities} to analyze an equivalent exponential mechanism $\algo'$ that uses the loss function $\mu'$ defined as follows: fix globally some $T_0 \in \trees(G)$ and define $\mu'(\vew, T) \coloneq \vew(T) - \vew(T_0)$.
	It holds that, for every fixed $\vew$, we can write $\mu'(\vew, T) = \mu(\vew, T) + c_\vew$, where $c_\vew = -\vew(T_0)$ does not depend on $T$, and thus $\algo$ and $\algo'$ are equivalent by \cref{lem:equivalent-utilities}.

	Let us choose the right $\lambda$ for $\algo'$ (and thus also for $\algo$). By the standard properties of the exponential mechanism (see \cref{lem:exp-mechanism-guarantees}), $\algo'$ is $\eps$-differentially private if we choose $\lambda \le \eps/(2\Delta)$, where $\Delta$ is the global sensitivity of $\mu'$, as defined in \cref{lem:exp-mechanism-guarantees}. As a next step, we bound $\Delta$.
	For each $\vew \siminf \vew'$, we have:
	\begin{align*}
	\left|\mu'(\vew, T) - \mu'(\vew', T)\right|
	&=
		\left|(\vew - \vew')(T) - (\vew - \vew')(T_0)\right|
	\\&=
	\bigg|\sum_{e \in T \setminus T_0} (\vew - \vew')(e) - \sum_{e \in T_0 \setminus T} (\vew - \vew')(e)\bigg|
	\\&\le
	\sum_{e \in T \setminus T_0} \left|(\vew - \vew')(e)\right| + \sum_{e \in T_0 \setminus T} \left|(\vew - \vew')(e)\right|
	\\&\le
	2\ham(T_0, T)
	\le 2R_0,
	\end{align*}
	for $R_0 \coloneqq \max_{T\in \trees(G)} \ham(T_0, T)$. We
	used the fact that edges present in both $T$ and $T_0$ do not count
	towards the result. Hence, $\Delta \le 2R_0 \le 2D$. If we thus choose $\lambda = \frac{\eps}{4R_0}$ in $\algo'$,
	we immediately get from \cref{lem:exp-mechanism-guarantees} that $\algo'$ is $\eps$-differentialy private and the expected error is at most $4R_0\log|\trees(G)|/\eps \le 4D\log|\trees(G)|/\eps$, as needed. By the equivalence of $\algo$ and $\algo'$, the same holds for $\algo$.

	Finally, note that $\algo$ can compute $R_0$ (and thus $\lambda$)
	quickly: namely, if we denote by $T^*$ the MST of a graph $(G, -\oneweights_{T_0})$, then $R_0 = \ham(T_0, T^*)$. That is because $T^*$ minimizes the expression $-\oneweights_{T_0}(T)$, which, by \cref{fact:oneweights-properties}, is equal to $-\ham(T_0, T)$, just as needed. The MST can be computed in linear time using e.g.~the Jarník-Prim algorithm with a double-ended queue as the priority queue, as all weights are either $-1$ or $0$.

	\rhinline{Přidáno:}

	Let us now deal with the $\simone$ case. $\algo$ will again be an instantiation of the exponential mechanism. This time, we set $\lambda = \eps / 2$ and analyze $\algo$ directly, without the use of \cref{lem:equivalent-utilities}. For any $\vew \simone \vew'$ and $T \in \trees(G)$, we immediately have $|\vew(T) - \vew'(T)| \le \|\vew - \vew'\|_1 \le 1$, and thus the global sensitivity of $\mu$ is $\Delta \le 1$. By \cref{lem:exp-mechanism-guarantees}, $\algo$ is $\eps$-differentially private and the expected error is at most $2 \log |\trees(G)| / \eps$, exactly as needed.
\end{proof}

Since $R_0$ computed in the above proof satisfies $D/2 \le R_0 \le D$, we
immediately get the following corollary:

\begin{corollary}
	A 2-approximation of $D = \diam(G)$ can be computed in linear time.
\end{corollary}

Note that the above algorithm is actually strictly stronger than \cref{alg:laplace}, as there are graphs where $\log |\trees(G)| = o(D\log n)$. We suspect that, in fact, the exponential mechanism is universally optimal for both $\simone$ and $\siminf$, but we were not able to prove a stronger lower bound.

\ifanonymous\else
\section*{Acknowledgements}
We would like to thank Rasmus Pagh for helpful discussions and hosting the first author at the University of Copenhagen. We would like to thank Bernhard Haeupler for helpful discussions.
\fi

\printbibliography

\end{document}